\documentclass[aps,pra,10pt,a4paper,reprint,nofootinbib,superscriptaddress]{revtex4-2}
\usepackage[utf8]{inputenc}
\usepackage{amsfonts}
\usepackage{amssymb}
\usepackage{amsmath}
\usepackage{amsthm}
\usepackage{graphics}
\usepackage{graphicx}
\usepackage{braket}
\usepackage[colorlinks=true,linkcolor=blue,urlcolor=blue,citecolor=blue]{hyperref}
\usepackage{times,txfonts}
\newtheorem{theorem}{Theorem}
\newtheorem{corollary}{Corollary}

\newtheorem{proposition}{Proposition}
\newtheorem{remark}{Remark}

\theoremstyle{definition}
\newtheorem{example}{Example}
\newcommand{\ketbra}[2]{|#1\rangle \langle #2|}

\newcommand{\alex}[1]{{\leavevmode\color{blue}[AS: #1]}}

\begin{document}
\title{Entanglement battery and entanglement catalyst in local state discrimination problems}

\author{Saronath Halder}
\affiliation{Department of Physics, School of Advanced Sciences, VIT-AP University, Beside AP Secretariat, Amaravati 522241, Andhra Pradesh, India}

\author{Aby Philip}
\affiliation{Institute of Fundamental Technological Research, Polish Academy of Sciences, Pawi\'{n}skiego 5B, 02-106 Warsaw,
Poland}

\author{Alexander Streltsov}
\affiliation{Institute of Fundamental Technological Research, Polish Academy of Sciences, Pawi\'{n}skiego 5B, 02-106 Warsaw,
Poland}

\begin{abstract}
In this work, we study the limitations and advantages of using entanglement battery and entanglement catalyst in local state discrimination problems. We consider both cases of such tools, i.e., exact and approximate cases. We show that to distinguish any set of orthogonal pure bipartite entangled states perfectly under local operations and classical communication using an (exact) entanglement battery or an (exact) entanglement catalyst, it is necessary to consider that the cardinality of the set must be smaller than the total dimension of the given Hilbert space. Then, we construct a nontrivial case where an exact entanglement battery can provide huge advantage. We also construct other nontrivial cases where exact or approximate entanglement battery or entanglement catalyst can be useful. In fact, we find that the approximate tools are particularly useful in the local discrimination of certain sets which can be derived from many-copy indistinguishable ensembles. 
\end{abstract}
\maketitle

\section{Introduction}\label{sec1}
One of the main goals in quantum information science is to explore the usefulness of quantum resources \cite{Chitambar19} in information processing protocols. For example, it is well known that quantum entanglement \cite{Horodecki09, Guhne09} is very useful as a resource in several protocols, such as quantum dense coding \cite{Bennett92, Mattle96}, quantum teleportation \cite{Bennett93, Bouwmeester97}, quantum cryptography \cite{Ekert91}. Because of this usefulness, a task is considered to be particularly important if after beginning this with additional entanglement as a resource, it is possible to get back some entanglement at the end along with achieving the goal. In Ref.~\cite{Jonathan99}, the authors showed that when manipulating pure entangled states locally, it is possible to construct a task where the entangled resource state can be recovered unaffected along with achieving the goal of the task. This is a process known as entanglement catalysis \cite{Datta23}. The concept of catalysts has come from chemistry, where they enable certain reactions or increase the rates of the reactions. However, in the whole process of a reaction, the catalyst remains unchanged. Almost similarly, in quantum information theory, an entanglement catalyst advances our capability to accomplish a task by local operations and classical communication (LOCC), provided that there is no loss of entanglement in the catalyst qualitatively or quantitatively \cite{Ganardi25}. In this work, we discuss the limitations and advantages of using entanglement catalysts in local state discrimination problems. We also consider a more general version of an entanglement catalyst which is known as an entanglement battery \cite{Ganardi25} in the context of local state discrimination problems.

The setting of a local state discrimination problem \cite{Bennett99-1, Bennett99, Walgate00, Ghosh01, Walgate02, Ghosh02, DiVincenzo03, Horodecki03, Ghosh04, Fan04, Horodecki04, Nathanson05, Watrous05, Niset06, Hayashi06, Feng09, Bandyopadhyay10, Bandyopadhyay11, Yu12, Yang13, Zhang14, Zhang15, Xu16, Halder18, Halder19, Banik21} is described in the following way. We consider a composite quantum system. This system is prepared in a state that is secretly chosen from a given set, $\mathcal{S}$. The task is to identify the state of the system by discriminating between the states of $\mathcal{S}$. However, there is a constraint. The subsystems are distributed among spatially separated locations. Now, because of this spatial separation, the parties of those locations are only allowed to use LOCC. Ultimately, the local state discrimination problem boils down to a task in which the parties of different locations apply LOCC and try to identify the state of the composite quantum system by discriminating between the states of $\mathcal{S}$. For our work, we assume that the states of $\mathcal{S}$ are pairwise orthogonal to each other. Therefore, the parties are required to identify the state of the composite quantum system perfectly. If the parties are unable to identify the state of the quantum system perfectly under LOCC, then we say that the given set $\mathcal{S}$ is a locally indistinguishable set. Now, we are ready to provide a general setup for the definitions of entanglement battery and entanglement catalyst in local state discrimination problems.

We begin with a locally indistinguishable set $\mathcal{S}$ = $\{\ket{\Psi_i}\}_i$. For this set, to accomplish the task of state discrimination under LOCC, it is natural to consume an additional entangled resource state $\tau$ (it can be a pure state or a mixed state). See Refs.~\cite{Cohen08, Bandyopadhyay09, Yu14, Bandyopadhyay16, Halder18, Bandyopadhyay18, Zhang18, Rout19, Shi20} in this context. 
Now, before applying LOCC and $\tau$, we have a locally indistinguishable set $\mathcal{S}$, and after applying LOCC and $\tau$, we have the index $\{i\}$, that is, we have the information of the state in which the given system is prepared. However, after completing the state discrimination task, we also assume that we get back an entangled resource state $\tau^\prime$ with probability $p~(>0)$, represented by $p\tau^\prime$. This process can be presented through an equation given below. 
\small
\begin{equation}\label{eq1}
\left\{\tau,~\mathcal{S}\right\}\xrightarrow{\mbox{LOCC}}\left\{p\tau^\prime,~\{i\}\right\}.
\end{equation}
\normalsize
Here, $\{\tau,~\mathcal{S}\}$ is the collection of inputs that we have before the local protocol. Similarly, $\{p\tau^\prime,~\{i\}\}$ is the collection of outputs that we have after the local protocol. 
Furthermore, we consider $t$ instances of the same state discrimination problem together, that is, $\mathcal{S}^{\otimes t}$. This is basically copies of a locally indistinguishable set and it is entirely different from the copies of quantum states. So, in this case, the process can be presented through the following equation:
\small
\begin{equation}\label{eq2}
\begin{array}{l}
\left\{\tau^{\otimes t},~\mathcal{S}^{\otimes t}\right\}\xrightarrow{\mbox{LOCC}}\\[1 ex]
\left\{\{p^t{\tau^\prime}^{\otimes t},~tp^{t-1}(1-p){\tau^\prime}^{\otimes t-1},\dots\}~\{i_1\},~\{i_2\},\dots,\{i_t\}\right\},
\end{array}
\end{equation}
\normalsize
where with probability $p^t$ we get ${\tau^\prime}^{\otimes t}$, with probability $tp^{t-1}(1-p)$ we get ${\tau^\prime}^{\otimes t-1}$, etc. Here, $t$ instances of a discrimination task are treated adaptively and the probabilities $p^t$, $tp^{t-1}(1-p)$, etc.~are from a binomial distribution. The indices $\{\{i_1\},~\{i_2\},\dots,\{i_t\}\}$ indicate identification of states in different instances. This is a general setup required for our purpose. Next, we define the tools one by one. 

(i) For (\ref{eq1}), if $p=1$ and $\tau=\tau^\prime$, then it is a case of an {\it exact entanglement catalyst}. We note here that $\tau=\tau^\prime$ means there is local equivalence between the states $\tau$ and $\tau^\prime$ (both conversions from $\tau$ to $\tau^\prime$ and from $\tau^\prime$ to $\tau$ is possible under LOCC). (ii) A more general version of an exact catalyst emerges when the state $\tau^\prime$ is a better state compared to $\tau$ (along with $p=1$), that is, local conversion from $\tau^\prime$ to $\tau$ is possible deterministically but the opposite conversion is not possible, a process known as an {\it exact entanglement battery} \cite{Ganardi25}. (iii) For (\ref{eq2}), when $0<p\leq1$ and $t\geq1$ if we get back $\tau$ on average or we get back $\tau$ with a probability very close to one, then it is a case of an {\it approximate entanglement catalyst}. (iv) For (\ref{eq2}), when $0<p\leq1$ and $t\geq1$ if we get back $\tau^\prime$ (which is a better state compared to $\tau$ as defined in the above) on average or we get back $\tau^\prime$ with a probability very close to one, then it is a case of an {\it approximate entanglement battery}, a more general version of an approximate entanglement catalyst. The approximate version of an entanglement catalyst is adopted from \cite{Kondra21, Bartosik21, Datta24}. Nevertheless, in the context of local state discrimination problems, the entanglement catalysts are not explored much. So far, only a few results have been reported \cite{Yu12, Cosentino14, Sen22}. In fact, there is no known instance of an entanglement battery (exact or approximate) in the context of local state discrimination problems. 

In the following, we briefly discuss about entanglement measures. Then, we describe our results of exact and approximate cases of entanglement catalysts and entanglement batteries in the local state discrimination problems.

\section{Entanglement Measures}
In the `Results' section, we will be using the concept `Entanglement Measures'. Therefore, here we briefly discuss about it. Suppose, two entangled states are given. It might be the case that the states are not equally entangled. But how to compare between the amounts of entanglement contained in the states. One way to do this is to quantify entanglement. If we consider a function $E$ which can be treated as an entanglement measure, then it must output different values for the states which are not equally entangled. Nevertheless, there are certain properties which must be satisfied by $E$ to be a good entanglement measure. (i) If $\rho$ is separable, then $E(\rho)=0$. (ii) For certain entangled states, $E$ must output a maximum value. These states are maximally entangled states. For all other states, $E$ must output a value between 0 and that maximum value. If we consider a two-qubit system, then $0\leq E\leq1$. We mention that the scaling of $E$ can be fixed in a different way as well. In this context, one can have a look into \cite{Vidal02}. (iii) $E$ must be non-increasing (on average) under LOCC. For example, $E\left[\Lambda_{\text{LOCC}}(\rho)\right]\leq E(\rho)$, where $\Lambda_{\text{LOCC}}$ is any operation which can be implemented by LOCC. (iv) $E$ can be additive on tensor product, i.e., $E(\rho\otimes\sigma) = E(\rho) + E(\sigma)$, see \cite{Christandl04} in this context. There can be other properties as well. In particular, if we consider multipartite systems, then the definitions become increasingly complex. So, for further details, we refer to \cite{Horodecki09, Guhne09}. Entropic measures are quite popular to use as entanglement measures, for example, von Neumann entropy \cite{Horodecki09}. Since, here we use this quantity, we provide its definition briefly.

We consider a bipartite pure state $\ket{\psi}$ (shared between two parties A and B). The reduced density matrices can be defined as the following.
\begin{equation*}
\rho_{x} = \text{Tr}_{\bar{x}}[\ket{\psi}\bra{\psi}],
\end{equation*}
where $x = A$ when $\bar{x} = B$ and vice-versa. Then, the von Neumann entropy can be defined as
\begin{equation*}
S(\rho_x) = - \text{Tr}[\rho_x\log_2\rho_x].
\end{equation*}
Following the same way, one can also provide the aforesaid definition of von Neumann entropy for a bipartite mixed state, $\rho$ instead of considering $\ket{\psi}$. In that case, the reduced density matrices are \(\rho_{x} = \text{Tr}_{\bar{x}}[\rho]\).

Another interesting entropic quantity is the quantum conditional entropy \cite{Horodecki09} which is often used in the context of quantum state merging \cite{Horodecki_2006} and to detect entanglement. The definition of quantum conditional entropy can be given using von Neumann entropy:
\begin{equation*}
S(\rho_{A|B}) = S(\rho) - S(\rho_B),
\end{equation*}
where for a bipartite state $\rho$, the conditional entropy is the entropy of the conditional density matrix $\rho_{A|B}$, represented by $S(\rho_{A|B})$. $S(\rho)$ is the von Neumann entropy of the whole state $\rho$ and $S(\rho_B)$ is the von Neumann entropy of the reduced density matrix, $\rho_B$ as defined earlier. Note that in other way the terms of the above equation are also represented as $S(\rho_{A|B})\equiv S(A|B)_\rho$, $S(\rho)\equiv S(AB)_\rho$, and $S(\rho_B)\equiv S(B)_\rho$. 

We are now ready to present our main findings.

\section{Results}\label{sec2}
We start with an observation that the entanglement batteries (exact) or the entanglement catalysts (exact) are not useful in all cases of local state discrimination problems (LSDPs). Such a case is outlined below.

\subsection{A no-go theorem}
\begin{theorem}\label{theo1}
It is impossible to distinguish an orthonormal basis $\mathcal{B}$ in $\mathcal{H} = \mathbb{C}^{d_1}\otimes\mathbb{C}^{d_2}$ (in brief, $d_1\otimes d_2$) perfectly under LOCC using an entanglement battery or an entanglement catalyst, $\tau$ if $\mathcal{B}$ contains only entangled states.

\end{theorem}

\begin{proof}
We first consider any entanglement catalyst, $\tau$ ($\tau$ can be a pure state or a mixed state). Now, following a proof technique of \cite{Horodecki03}, we start with a state $\ket{\Psi}$, given by-
\begin{equation}\label{eq3}
\ket{\Psi} = \left(\frac{1}{\sqrt{d_1}}\sum_{i=1}^{d_1}\ket{ii}\right)_{AC}\otimes\left(\frac{1}{\sqrt{d_2}}\sum_{j=1}^{d_2}\ket{jj}\right)_{BD}
\end{equation}
where $A,~B,~C,~D$ are four spatially separated parties. Clearly, the state $\ket{\Psi}$ is a product state in $AC:BD$ bipartition. We rewrite the state $\ket{\Psi}$ in $AB:CD$ bipartition and it will become:
\begin{equation*}
\ket{\Psi} = \left(\frac{1}{\sqrt{d_1d_2}}\sum_{k=1}^{d_1d_2}\ket{k}_{AB}\ket{k}_{CD}\right).    
\end{equation*}
Next, we consider $\mathcal{B}\equiv\{\ket{\psi_k}\}_{k=1}^{d_1d_2}$ is an orthonormal basis in $d_1\otimes d_2$. We consider a unitary operator $U$, such that $U\ket{k} = \ket{\psi_k}$ $\forall k = 1,2,\dots,d_1d_2$. Then, we use $U\otimes U^{\ast}$ invariance of $\ket{\Psi}$ in $AB:CD$ bipartition and we further rewrite $\ket{\Psi}$ as $\frac{1}{\sqrt{d_1d_2}}(\sum_{k=1}^{d_1d_2}\ket{\psi_k}_{AB}\ket{\psi_k}^{\ast}_{CD})$, where $\ket{\psi_k}^\ast$ is the complex conjugate of $\ket{\psi_k}$ in the computational basis. So, if $\mathcal{B}$ contains only entangled states and they are locally distinguishable, then the states $\{\ket{\psi_k}_{AB}\}_k$ can be distinguished locally, releasing an entangled state $\ket{\psi_k}^\ast_{CD}$. This creates entanglement across $AC:BD$ bipartition. But recall that we have started with $\ket{\Psi}$ which was a product state in the $AC:BD$ bipartition. Clearly, distinguishing the states $\{\ket{\psi_k}_{AB}\}_k$ under LOCC is forbidden when they are entangled because creating entanglement from product state under LOCC is not possible. This proves that any entangled basis is locally indistinguishable.

However, we now assume that there is a catalyst-assisted local protocol that distinguishes the states $\{\ket{\psi_k}_{AB}\}_k$ perfectly. So, when such a protocol is acted on $\ket{\Psi}$ and distinguishes the states $\{\ket{\psi_k}_{AB}\}_k$, $\ket{\psi_k}^\ast_{CD}$ can be released along with the catalyst state. Overall, we can have the following transformation: 
\begin{equation}\label{eq4}
\Lambda_{\mbox{LOCC}}\left[\ketbra{\Psi}{\Psi}\otimes\tau\right] = \left(\ketbra{\psi_k}{\psi_k}\right)^\ast_{CD}\otimes\tau.
\end{equation}
Since LOCC cannot increase entanglement, we can assume $E[\ketbra{\Psi}{\Psi}\otimes\tau]\geq E[(\ketbra{\psi_k}{\psi_k})_{CD}^\ast\otimes\tau]$; $E$ is an entanglement measure. We further assume two things: (a) we are sticking to a particular bipartition $AC:BD$, and (b) $E$ is additive on tensor products like ``Squashed Entanglement'' \cite{Christandl04}. Thus, we obtain $E(\ket{\Psi}_{AC:BD})+E(\tau)\geq E(\ket{\psi_k}^\ast_{CD})+E(\tau)$. $E(\ket{\Psi}_{AC:BD})$ = 0 due to separability of the state in $AC:BD$ bipartition and $E(\ket{\psi_k}^\ast_{CD})>0$ when the considered basis, $\mathcal{B}$ contains only entangled states. Clearly, $E(\ket{\psi_k}^\ast_{CD})>E(\ket{\Psi}_{AC:BD})$. Or, $E(\ket{\psi_k}^\ast_{CD})+E(\tau)>E(\ket{\Psi}_{AC:BD})+E(\tau)$. Thus, $E(\ket{\Psi}_{AC:BD})+E(\tau)\geq E(\ket{\psi_k}^\ast_{CD})+E(\tau)$ is impossible. This implies that (\ref{eq4}) and corresponding assumption is not true. Therefore, perfect local discrimination of any orthonormal entangled basis, $\mathcal{B}\equiv\{\ket{\psi_k}\}_k$ under LOCC using an entanglement catalyst, $\tau$ is not possible. 

From this proof technique, it becomes obvious that there is no entanglement battery as well using which it is possible to distinguish the states of $\mathcal{B}$ perfectly under LOCC. This is due to the following fact. In case of an entanglement battery, $\tau$ the transformation is given by
\begin{equation*}
\Lambda_{\mbox{LOCC}}\left[\ketbra{\Psi}{\Psi}\otimes\tau\right] = \left(\ketbra{\psi_k}{\psi_k}\right)^\ast_{CD}\otimes\tau^\prime,
\end{equation*}
where $E(\tau^\prime)\geq E(\tau)$ for some entanglement measure $E$. Ultimately, we get an equation: $E(\ket{\Psi}_{AC:BD})+E(\tau)\geq E(\ket{\psi_k}^\ast_{CD})+E(\tau^\prime)$ for squashed entanglement which is required to be satisfied to achieve our goal. However, this is impossible as it is clearly evident from the above. This completes the proof.
\end{proof}

An immediate corollary due to the above theorem is given as the following.
\begin{corollary}\label{coro1}
We assume that for a locally indistinguishable set $\mathcal{S}$ of pure orthogonal entangled states, there exists a catalyst-assisted local protocol using which perfect discrimination of the states of $\mathcal{S}$ is possible. Then, for such a perfect discrimination, it is necessary that $|\mathcal{S}|<D$, where $|\mathcal{S}|$ is the cardinality of the set $\mathcal{S}$ and $D$ is the total dimension of the given Hilbert space.     
\end{corollary}
This is straightforward from Theorem \ref{theo1} because if $|\mathcal{S}|=D$, then the set forms an orthonormal basis in the given Hilbert space which contains only entangled states. 

Nevertheless, we can also state the following corollary.
\begin{corollary}\label{coro2}
We consider any orthonormal multipartite basis $\mathcal{G}$ which contains only genuinely entangled states. For the perfect local discrimination of the states of $\mathcal{G}$, it is not possible to find a bipartite entangled state which can serve as an entanglement battery and releases a genuine multipartite entangled state after the completion of the discrimination task.
\end{corollary}
To prove the above, one can consider a bipartition in which the entanglement battery is not shared but in that bipartition entanglement is generated after the completion of the discrimination task. Since, the entanglement battery is a bipartite one but a genunine multipartite entangled state is produced in the process, the aforesaid bipartition always exists. The reason is that the bipartite entanglement is shared between two parties only, but the genuine multipartite entanglement is shared among at least three parties. Thereafter, the bipartite proof technique works starting from (\ref{eq3}). 

From the discussions so far (Theorem \ref{theo1}, Corollary \ref{coro1}, and Corollary \ref{coro2}), it is quite clear that constructions of specific cases of LSDPs where the given sets of entangled states do not form an orthonormal basis are particularly important because for them an exact entanglement battery or an exact entanglement catalyst may provide an advantage. We now proceed to present such constructions.


\subsection{Exact cases}
We consider the following multipartite set of four orthogonal quantum states. [From now on, we consider only pure state entanglement as an entanglement battery or an entanglement catalyst. Also, we may avoid the normalization coefficients in several places because in those places the coefficients do not play any important role.]
\begin{equation}\label{eq5}
\begin{array}{c}
\ket{\Psi_0} = \ket{\psi_0}\otimes\ket{\psi_0^\prime},~\ket{\Psi_1} = \ket{\psi_1}\otimes\ket{\psi_1^\prime},\\[1 ex]
\ket{\Psi_2} = \ket{\psi_2}\otimes\ket{\psi_1^\prime},~\ket{\Psi_3} = \ket{\psi_3}\otimes\ket{\psi_1^\prime}.
\end{array}
\end{equation}
The states $\ket{\psi_i}$ ($\forall~i=0,1,2,3$) are $n$-qubit states and $\ket{\psi_i^\prime}$ ($\forall~i=0,1$) are also $n$-qubit states. Clearly, $\ket{\Psi_i}$ are $2n$-qubit states ($\forall~i=0,1,2,3$). $\ket{\Psi_i}$ are shared between $n$ spatially separated parties. Corresponding Hilbert space is $\mathcal{H}=[\mathbb{C}^2\otimes\mathbb{C}^2]^{\otimes n}$. So, each party holds a two-qubit system. The states $\{\ket{\psi_i}\}$ are given as the following.
\begin{equation}\label{eq6}
\begin{array}{c}
\ket{\psi_0} = \ket{0}^{\otimes k}\ket{0}^{\otimes n-k} + \ket{1}^{\otimes k}\ket{1}^{\otimes n-k},\\[1 ex]

\ket{\psi_1} = \ket{0}^{\otimes k}\ket{0}^{\otimes n-k} - \ket{1}^{\otimes k}\ket{1}^{\otimes n-k},\\[1 ex]

\ket{\psi_2} = \ket{0}^{\otimes k}\ket{1}^{\otimes n-k} + \ket{1}^{\otimes k}\ket{0}^{\otimes n-k},\\[1 ex]

\ket{\psi_3} = \ket{0}^{\otimes k}\ket{1}^{\otimes n-k} - \ket{1}^{\otimes k}\ket{0}^{\otimes n-k},
\end{array}
\end{equation}
where $1\leq k<n$. We note here that the state $\ket{\psi_0^\prime}$ has the same form as $\ket{\psi_0}$. Similarly, the state $\ket{\psi_1^\prime}$ has the same form as $\ket{\psi_1}$. We now prove the following.

\begin{proposition}\label{prop2}
The set $\{\ket{\Psi_i}\}_{i=0}^3$ is locally indistinguishable. A $2\otimes2$ maximally entangled state as resource is sufficient to distinguish the states of the given set under LOCC. Furthermore, in this process of discrimination, it is possible to get back a $2^{\otimes n}$ Greenberger-Horne-Zeilinger (GHZ) state.
\end{proposition}

\begin{proof}
The proof consists of a couple of steps. In the following, these steps are presented one by one.

(i) First, we prove that the states of (\ref{eq5}) are locally indistinguishable. For this purpose, we recall the forms of $\{\ket{\psi_i}\}_{i=0}^3$ and group $n$ qubits into two: first $k$ qubits and remaining $n-k$ qubits. Thus, we find a bipartition of $k$-qubit versus $n-k$ qubits. Then, we consider the following mapping: $\ket{0}^{\otimes k}\rightarrow\ket{\mathbf{0}}$, $\ket{0}^{\otimes n-k}\rightarrow\ket{\tilde{\mathbf{0}}}$, $\ket{1}^{\otimes k}\rightarrow\ket{\mathbf{1}}$, $\ket{1}^{\otimes n-k}\rightarrow\ket{\tilde{\mathbf{1}}}$. We apply similar mapping for $\{\ket{\psi_i^\prime}\}_{i=0}^1$. So, the set of (\ref{eq5}) is mapped into the following:
\begin{equation}\label{eq7}
\begin{array}{l}
\ket{\Psi_0} = (\ket{\mathbf{0}\tilde{\mathbf{0}}} + \ket{\mathbf{1}\tilde{\mathbf{1}}})\otimes(\ket{\mathbf{0}\tilde{\mathbf{0}}} + \ket{\mathbf{1}\tilde{\mathbf{1}}}),\\[1 ex]
\ket{\Psi_1} = (\ket{\mathbf{0}\tilde{\mathbf{0}}} - \ket{\mathbf{1}\tilde{\mathbf{1}}})\otimes(\ket{\mathbf{0}\tilde{\mathbf{0}}} - \ket{\mathbf{1}\tilde{\mathbf{1}}}),\\[1 ex]
\ket{\Psi_2} = (\ket{\mathbf{0}\tilde{\mathbf{1}}} + \ket{\mathbf{1}\tilde{\mathbf{0}}})\otimes(\ket{\mathbf{0}\tilde{\mathbf{0}}} - \ket{\mathbf{1}\tilde{\mathbf{1}}}),\\[1 ex]
\ket{\Psi_3} = (\ket{\mathbf{0}\tilde{\mathbf{1}}} - \ket{\mathbf{1}\tilde{\mathbf{0}}})\otimes(\ket{\mathbf{0}\tilde{\mathbf{0}}} - \ket{\mathbf{1}\tilde{\mathbf{1}}}).
\end{array}
\end{equation}
These states look like two Bell states in tensor product form. It is known that states in such specific forms cannot be perfectly distinguished by LOCC \cite{Yu12}. So, the states of (\ref{eq5}) are locally indistinguishable in a particular bipartition. Finally, we note that a multipartite set is locally indistinguishable if it is locally indistinguishable in a bipartition. This completes the proof of local indistinguishability.

(ii) Next, we consider a minimum dimensional resource state $\ket{00}+\ket{11}$ as $\ket{\xi}$. Using this resource state the parties try to distinguish the states $\{\ket{\psi_i}\}_{i=0}^3$ without disturbing the states $\{\ket{\psi_0^\prime},~\ket{\psi_1^\prime}\}$. The resource state is shared between two parties. The first party is picked from the group of parties among which the first $k$ qubits are distributed and the second party is picked from the group of parties among which the remaining $n-k$ qubits are distributed. Now, after sharing $\ket{\xi}$, one of these two parties teleports a subsystem (a qubit) to the location of the other party. In such a situation, the set of (\ref{eq5}) is transformed to the following: 
\small
\begin{equation}\label{eq8}
\begin{array}{l}
\ket{\Psi_0} = (\ket{0}^{\otimes k-1}\ket{00}\ket{0}^{\otimes n-k-1} + \ket{1}^{\otimes k-1}\ket{11}\ket{1}^{\otimes n-k-1})\otimes\ket{\psi_0^\prime},\\[1 ex]
\ket{\Psi_1} = (\ket{0}^{\otimes k-1}\ket{00}\ket{0}^{\otimes n-k-1} - \ket{1}^{\otimes k-1}\ket{11}\ket{1}^{\otimes n-k-1})\otimes\ket{\psi_1^\prime},\\[1 ex]
\ket{\Psi_2} = (\ket{0}^{\otimes k-1}\ket{01}\ket{1}^{\otimes n-k-1} + \ket{1}^{\otimes k-1}\ket{10}\ket{0}^{\otimes n-k-1})\otimes\ket{\psi_1^\prime},\\[1 ex]
\ket{\Psi_3} = (\ket{0}^{\otimes k-1}\ket{01}\ket{1}^{\otimes n-k-1} - \ket{1}^{\otimes k-1}\ket{10}\ket{0}^{\otimes n-k-1})\otimes\ket{\psi_1^\prime}.
\end{array}
\end{equation}
\normalsize
Now, one party has two qubits (of $\ket{\psi_i}$) after the teleportation. This is shown in a separate ket notation in the above equation. Next, this party who has two qubits (of $\ket{\psi_i}$), performs a measurement on these qubits. This measurement is defined by two projection operators: $P_1 = |00\rangle\langle00|+|11\rangle\langle11|$ and $P_2 = |01\rangle\langle01|+|10\rangle\langle10|$. After such a measurement, two states from the set $\{\ket{\psi_i}\}_{i=0}^3$ get eliminated. Since, the aforesaid measurement is orthogonality-preserving, the remaining two states (orthogonal and pure states) can also be distinguished by LOCC \cite{Walgate00}. In this way, after local discrimination among the states of the set $\{\ket{\psi_i}\}_{i=0}^3$, the $n$ spatially separated parties are either left with $\ket{\psi_0^\prime}$ or $\ket{\psi_1^\prime}$, an $n$-qubit GHZ state. This completes the proof of the above proposition. 
\end{proof}

In the above, we have seen that with the help of a two-qubit maximally entangled state $\ket{\xi}$, it is possible to distinguish the states of (\ref{eq5}). Furthermore, $n$ parties are left with a genuinely entangled state with unit probability. This state is basically $\ket{\xi^\prime}$ which can never be found from $\ket{\xi}$ under LOCC. This is because it is never possible to obtain a genuinely entangled state locally from a bipartite state, not even probabilistically. But the opposite conversion might be possible. For example, here it is possible to obtain the bipartite entangled state $\ket{\xi}$ locally form the genuinely entangled state $\ket{\xi^\prime}$. So, qualitatively, $\ket{\xi^\prime}$ is a better state compared to the state $\ket{\xi}$. This is an example of an exact entanglement battery assisted local state discrimination, a multi-qubit generalization of an exact entanglement catalyst assisted local state discrimination, given in \cite{Yu12}. However, what we obtain here is much stronger and this is due to the following. 
\begin{remark}\label{rem1}
The resource state $\ket{\xi}$ is always a $2\otimes2$ maximally entangled state here. 
On the other hand, the state $\ket{\xi^\prime}$ which is received after the discrimination process is a $2^{\otimes n}$ maximally entangled GHZ state where `$n$' can be very large. Clearly, we obtain a huge advantage in recovering entanglement.   
\end{remark}
Next, we construct an exact case of multi-qubit entanglement catalyst. For this purpose we consider, the states $\{\ket{\psi_i}\otimes\ket{\psi_1^\prime}\}_{i=4}^{2^n-1}$ along with the states of (\ref{eq5}) and for this new set, we present the following proposition.
\begin{proposition}\label{prop3}
A $2^{\otimes n}$ maximally entangled GHZ state acts as an exact entanglement catalyst in the local discrimination of $\{\ket{\Psi_i}\}_{i=0}^{2^n-1}$, where $\ket{\Psi_0}$ = $\ket{\psi_0}\otimes\ket{\psi_0^\prime}$, $\ket{\Psi_i}$ = $\ket{\psi_i}\otimes\ket{\psi_1^\prime}$ for $i=1,\dots,2^n-1$, $\{\ket{\psi_i}\}_{i=0}^{2^n-1}$ constitutes a $2^{\otimes n}$ maximally entangled GHZ basis, $\{\ket{\psi_i}\}_{i=0}^3$ are defined in (\ref{eq6}) and $\{\ket{\psi_0^\prime},~\ket{\psi_1^\prime}\}$ having the same form as that of $\{\ket{\psi_0},~\ket{\psi_1}\}$ respectively. 
\end{proposition}

\noindent
{\bf Note:} $\forall i =0,1,\dots,2^n-1$, $\ket{\Psi_i}$ are $2n$-qubit states but the states are shared among $n$ spatially separated parties. So, each party holds a two-qubit system like the previous case.

\begin{proof}
The set $\{\ket{\Psi_i}\}_{i=0}^{2^n-1}$ must be locally indistinguishable because it contains a locally indistinguishable subset $\{\ket{\Psi_i}\}_{i=0}^{3}$, proved in Proposition \ref{prop2}. Next, we distinguish the states $\{\ket{\psi_i}\}_{i=0}^{2^n-1}$ without disturbing $\ket{\psi_0^\prime}$ and $\ket{\psi_1^\prime}\}$ with the help of an $n$-qubit maximally entangled GHZ state. It is known that an $n$-qubit maximally entangled GHZ basis can be perfectly distinguished 
with the help of an $n$-qubit maximally entangled GHZ state as resource~\cite{Bandyopadhyay18}. Now after discrimination, the spatially separated parties are left with either $\ket{\psi_0^\prime}$ or $\ket{\psi_1^\prime}$. This completes the proof.
\end{proof}

As seen above, teleportation or more generally the state merging protocol~\cite{Horodecki_2006} can be a useful tool for the task of local state discrimination. Specifically, Alice could use state merging, expending an appropriate amount of additional entanglement between Alice and Bob, to send system $A$ to Bob, who stores the state in system $\tilde{B}$. Then Bob can perform the discrimination task locally on $B$ and $\tilde{B}$ and announce the results. With the remaining system $A'$ and $B'$, Alice and Bob can try to recover the entanglement they expended to perform the initial state merging.

Suppose that we have two pure states, $\psi_{AB}$, and $\xi_{A'B'}$. Here the entanglement cost of state merging $\psi_{AB}\otimes\xi_{A'B'}$ is $S(A\vert B)_{\psi}+S(A'\vert B')_{\xi}$ (in terms of conditional entropy) and entanglement content in $\xi_{A'B'}$ is $S(A')_{\xi}$ (in terms of von Neumann entropy). We want to use the state merging protocol to send $AA'$ to Bob such that, using the ebits obtained as a result of the state merging protocol, we can distill $\xi_{A'B'}$. For this, we require that the entanglement content of $\xi_{A'B'}$ be less than the entanglement gained as result of state merging i.e., $-S(A\vert B)_{\psi}-S(A'\vert B')_{\xi}$, or 
\begin{align}
	&-S(A\vert B)_{\psi}-S(A'\vert B')_{\xi}> S(A')_{\xi}\nonumber\\
	&\implies S(B')_{\xi}-S(A'B')_{\xi} -S(A\vert B)_{\psi} > S(A')_{\xi}\nonumber\\
	&\implies -S(A'B')_{\xi} +S(B')_{\xi} -S(A')_{\xi} > S(A\vert B)_{\psi}\nonumber\\
	&\implies -S(A'B')_{\xi} > S(A\vert B)_{\psi}
\end{align}
Since $\xi_{A'B'}$ is a pure state, we get,
\begin{equation}
	S(A\vert B)_{\psi}< 0 \implies S(A)_{\psi}>0
\end{equation}
This implies that in order to recover $\xi_{A'B'}$, after sending $AA'$ to Bob using the state merging, $\psi_{AB}$ must be itself be entangled. Therefore, a set of pure bipartite states $\{\psi_{AB}^{i}\}$ can be perfectly locally distinguished using an entangled catalyst and state merging only if all $\psi_{AB}^{i}$ are entangled.

Next, we discuss some cases of approximate entanglement battery and entanglement catalyst.

\subsection{Approximate cases}
We consider a bipartite quantum system associated with the Hilbert space $\mathcal{H}$ = $\mathbb{C}^d\otimes\mathbb{C}^d$ (in brief, we say $d\otimes d$ instead of $\mathbb{C}^d\otimes\mathbb{C}^d$). For such a Hilbert space, we prove a proposition outlined below.
\begin{proposition}\label{prop4}
If $d = {d^\prime}^2$ and $d$ is sufficiently large, then there exists a set of $d+1$ maximally entangled states in $d\otimes d$, for local discrimination of which a $d^\prime\otimes d^\prime$ maximally entangled state can be useful as an approximate entanglement catalyst.
\end{proposition}

\begin{proof}
We begin the proof by noting that in $d\otimes d$, $d+1$ maximally entangled states cannot be perfectly distinguished by LOCC \cite{Ghosh04, Hayashi06}. Then, we provide the structure of the maximally entangled states in $d\otimes d$, given below:
\begin{equation}\label{eq9}
\ket{\Psi_0} = \ket{\psi_0}\otimes\ket{\psi^\prime_0},~ \ket{\Psi_i} = \ket{\psi_j}\otimes\ket{\psi^\prime_1},    
\end{equation}
where $i=1,\dots,d$ and $j=i-1$, so, $j=0,1,\dots,{d^\prime}^2-1$. The states $\{\ket{\psi_j}\}$, for $j=0,1,\dots,{d^\prime}^2-1$ form an orthogonal maximally entangled basis in $d^\prime\otimes d^\prime$. The states $\ket{\psi^\prime_0}$ and $\ket{\psi^\prime_1}$ have the same form as $\ket{\psi_0}$ and $\ket{\psi_1}$ respectively. The states $\ket{\Psi_i}$ are equally probable $\forall i=0,1,\dots,d$.

Now, to distinguish the states $\{\ket{\Psi_i}\}_{i=0}^d$ locally, a $d^\prime\otimes d^\prime$ maximally entangled state is provided as resource. Using this resource state, we distinguish among the basis states $\{\ket{\psi_j}\}$ without disturbing $\ket{\psi^\prime_0}$ or $\ket{\psi^\prime_1}$ via a teleportation based protocol. If in the discrimination process the states $\{\ket{\psi_j}\}$ are identified for $j=1,\dots,{d^\prime}^2-1$ then corresponding states are $\ket{\Psi_i}$, $i=2,\dots,d$. In this case, the parties are also left with $\ket{\psi^\prime_1}$, a $d^\prime\otimes d^\prime$ maximally entangled state. However, the problem occurs when the identified state in the discrimination process is $\ket{\psi_0}$ as it corresponds to two states $\ket{\Psi_0}$ and $\ket{\Psi_1}$. Here further discrimination can only be done by distinguishing between two orthogonal pure states $\ket{\psi^\prime_0}$ and $\ket{\psi_1^\prime}$ which can be done locally \cite{Walgate00}. In this case, no entangled state is left.

Overall, we find two possibilities. (i) With probability $\frac{2}{d+1}$, we do not complete the state discrimination task, i.e., when we are left with $\ket{\psi_0}$, we stop further discrimination. In this case, either we are left with $\ket{\psi^\prime_0}$ or $\ket{\psi_1^\prime}$, a $d^\prime\otimes d^\prime$ maximally entangled state. Actually, we do not need additional entanglement at all for availing this option. (ii) Otherwise, we can complete the state discrimination task (which is of course our priority here) and in this case, we are left with $\ket{\psi_1^\prime}$ with probability $1-\frac{2}{d+1}$. When $d$ is sufficiently large, the factor $\frac{2}{d+1}$ can be very close to $0$. In that case, we are left with a $d^\prime\otimes d^\prime$ maximally entangled state and corresponding probability can be arbitrary close to $1$. This completes the proof.
\end{proof}

In the context of the above proposition, we mention that the motivation of the aforesaid approximate case has come from embezzling effect \cite{Dam03} (also see \cite{Xing24}). Nevertheless, we note here a very basic difference between what we have found and the results reported in~\cite{Dam03,Xing24}. In particular, in our case we recover the catalyst state with high probability, very close to unity but not exactly unity (with the form of the catalyst state unchanged). On the other hand, in those papers, the recovery of an entangled state is discussed, where the recovered state is very close to the original catalyst state. We also mention it clearly that in our case the actual embezzling effect does not occur because we do not allow any modification of the original entangled state. Even if in a process, we get back a different state, that state must be converted deterministically to the original catalyst state via some local protocol.

We now consider sets that are locally indistinguishable in many-copy scenario. The number of copies is finite here. Such sets can be found in \cite{Bandyopadhyay11, Yu14, Li17, Halder24}. We simply say these sets as many-copy indistinguishable sets. We consider such a set $\{\rho_1, \rho_2\}$. Then, we construct a two-element set $\{\sigma_1, \sigma_2\}$, where $\sigma_1 = \rho_1^{\otimes m}$, $\sigma_2 = \rho_2^{\otimes m}$, and $m$ is the number of copies which is finite. Both states $\sigma_1$ and $\sigma_2$ are equally probable. In fact, one of these states is entangled and the other is not. Now for the set $\{\sigma_1, \sigma_2\}$, we present the following proposition.

\begin{proposition}\label{prop5}
Let us assume that the entangled resource state is $\ket{\xi}$ which is used in the local discrimination of $\{\sigma_1, \sigma_2\}$. If it is possible to get back the state $\ket{\xi}$ on average after the discrimination process, then $\ket{\xi}$ serves as an approximate entanglement catalyst.
\end{proposition}

\begin{proof}
Notice that the set $\{\sigma_1, \sigma_2\}\equiv\{\rho_1^{\otimes m}, \rho_2^{\otimes m}\}$ is locally indistinguishable. This follows from the fact that $\{\rho_1, \rho_2\}$ is locally indistinguishable even if many copies of the states $\rho_1$ and $\rho_2$ are provided. Here $m$ is finite.

The discrimination process is quite similar to the previous protocols. We write $\{\sigma_1, \sigma_2\}$ as $\{\rho_1\otimes\rho_1^{\otimes m-1}, \rho_2\otimes\rho_2^{\otimes m-1}\}$ and simply distinguish between $\rho_1$ and $\rho_2$ via a local protocol with the help of $\ket{\xi}$. This discrimination process should be done in such a way that the parts $\rho_1^{\otimes m-1}$, $\rho_2^{\otimes m-1}$ are not disturbed. So, after the discrimination process, $m-1$ copies of the states $\rho_1$ or $\rho_2$ are left untouched. Now, if one of the states is entangled, then with probability $\frac{1}{2}$, it is possible to gain $m-1$ copies of the entangled state. If these copies can be locally converted to $\ket{\xi}^{\otimes 2}$ deterministically, then we get back the entanglement on average considering any additive measure of entanglement. In particular, the average entanglement gained is $\frac{1}{2}E(\ket{\xi}^{\otimes 2})$ = $\frac{1}{2}[E(\ket{\xi}) + E(\ket{\xi})]$ = $E(\ket{\xi})$; $E$ is an additive measure of entanglement. This completes the proof.
\end{proof}

The above proposition is sufficiently general as it can be applied to any bipartite or multipartite Hilbert spaces. Furthermore, the above proposition can also be treated as an application of many-copy indistinguishable sets. Let us now take explicit examples which fit into the above proposition. 

\begin{example}\label{ex1}
We assume that $\rho_1 = |\Phi\rangle\langle\Phi|$, $\ket{\Phi} = \ket{00}+\ket{11}+\cdots+\ket{d-1~d-1}$, a two-qudit maximally entangled state. We also assume that $\rho_2 = \mathbb{I}-|\Phi\rangle\langle\Phi|$, $\mathbb{I}$ is the identity operator acting on the two-qudit Hilbert space. Then, we can fix $m\geq3$. For simplicity, we construct $\{\sigma_1, \sigma_2\}\equiv\{\rho_1^{\otimes3}, \rho_2^{\otimes3}\}$ and we fix $\ket{\xi}$ as $\ket{00}+\ket{11}+\cdots+\ket{d-1~d-1}$, a two-qudit maximally entangled state. It is known that $\{\sigma_1, \sigma_2\}$ is locally indistinguishable \cite{Yu14}. Here, the discrimination protocol can simply be a teleportation-based protocol. After discrimination between $\rho_1$ and $\rho_2$ without disturbing the parts ${\rho_1}^{\otimes2}$, ${\rho_2}^{\otimes2}$, the parties are left with $\rho_1^{\otimes2}$ with probability $\frac{1}{2}$. Interestingly, here the number of remaining copies can be quite larger than two when $m>>3$. Still, the resource state $\ket{\xi}$ can be a two-qudit maximally entangled state. This can be regarded as an approximate case of the entanglement battery.
\end{example}

In the above example, the state $\ket{\Phi}$ may not be a maximally entangled state. Still, we can demonstrate catalytic processes. Let us proceed to explain that scenario. We assume $\ket{\Phi}$ = $\alpha_0\ket{00}+\alpha_1\ket{11}+\cdots+\alpha_{d-1}\ket{d-1~d-1}$; $\alpha_i>0$, $\forall i = 0,1,\dots,d-1$, $\alpha_0\geq\alpha_1\geq\cdots\geq\alpha_{d-1}$ and $\alpha_0^2+\alpha_1^2+\cdots+\alpha_{d-1}^2=1$. Then, we consider `$t$' instances of state discrimination problems together, i.e., we consider $\mathcal{S}^{\otimes t}$ $\equiv$ $\{\sigma_1, \sigma_2\}^{\otimes t}$ $\equiv$ $\{\rho_1^{\otimes m}, \rho_2^{\otimes m}\}^{\otimes t}$; $\rho_1 = \ket{\Phi}\bra{\Phi}$, $\rho_2 = \mathbb{I}-\ket{\Phi}\bra{\Phi}$, $\mathbb{I}$ is the identity operator, acting on a two-qudit Hilbert space. In single state discrimination problem, i.e., when $t=1$, we get back $\rho_1^{\otimes m}$ with probability $\frac{1}{2}$. Here the goal is to produce two copies of two-qudit maximally entangled state from $\rho_1^{\otimes m}$ via a local strategy. Corresponding optimal probability can be found from \cite{Vidal99}. We want this probability to be one. So, we apply Nielsen's condition \cite{Nielsen99}. Thus, we get a bound on the entanglement content of $\ket{\Phi}$ by the equation given below:
\begin{equation}\label{eq10}
\alpha_{0} \leq \left(\frac{1}{d}\right)^{\frac{1}{m}}.
\end{equation}
Next, we consider $t>1$. In this case, we consider the distribution given in (\ref{eq2}). We may consider $p=\frac{1}{2}$ and $\ket{\xi^\prime}$ = $\rho_1^{\otimes m}$. We also assume that (\ref{eq10}) is still satisfied. So, the total amount of entanglement which can be extracted via a local strategy is given by $tE(\ket{\xi})$; where $E$ is an additive measure of entanglement. This is also an approximate case of the entanglement battery.

\begin{example}\label{ex2}
We now consider an example of a multi-qubit system. We assume that $\rho_1 = |G\rangle\langle G|$, $\ket{G} = \ket{0}^{\otimes n}+\ket{1}^{\otimes n}$, an $n$-qubit GHZ state. We also assume that $\rho_2 = \mathbb{I}-|G\rangle\langle G|$, $\mathbb{I}$ is the identity operator acting on the $n$-qubit Hilbert space. Then, we can fix $m \geq 3$. For simplicity, we construct $\{\sigma_1, \sigma_2\}\equiv\{\rho_1^{\otimes3}, \rho_2^{\otimes3}\}$. Here we fix $\ket{\xi}$ as $\ket{0}^{\otimes n}+\ket{1}^{\otimes n}$. $\{\rho_1, \rho_2\}$ and thereby $\{\sigma_1, \sigma_2\}$ is locally indistinguishable across every bipartition, according to Ref.~\cite{Yu14}. Here, the discrimination protocol is not a teleportation-based protocol. In fact, here the discrimination protocol can be derived from \cite{Bandyopadhyay18}. Actually, a resource state $\ket{0}^{\otimes n}+\ket{1}^{\otimes n}$ is sufficient to distinguish an $n$-qubit GHZ basis \cite{Bandyopadhyay18}. This is equivalent to the fact that the resource state $\ket{0}^{\otimes n}+\ket{1}^{\otimes n}$ is also sufficient to distinguish $\rho_1$ and $\rho_2$. However, after discrimination, the parties are left with $\rho_1^{\otimes2}$ with probability $\frac{1}{2}$. Here also the number of remaining copies can be quite larger than two. Still, the resource state $\ket{\xi}$ can be an $n$-qubit GHZ state. We note that here we need a measure of multipartite entanglement, which is additive for GHZ states, to demonstrate the catalytic effect. Such a measure can be found in \cite{Rubboli24}.
\end{example}
 
\begin{example}\label{ex3}
So far, in the examples corresponding to Proposition \ref{prop5}, we see that a bipartite (or, multipartite) entangled state is used as a resource and then multiple copies of a bipartite (or, multipartite) entangled state are recovered. However, we are now going to discuss an interesting three-qubit example where we use a bipartite entangled state as a resource but after state discrimination multiple copies of a multipartite entangled state are recovered. This is the best possible case that we have found to demonstrate the approximate multipartite entanglement battery. For this purpose, we consider a three qubit unextendible product basis (UPB)\footnote{Given a set of fully product states, if they span a proper subspace of the considered Hilbert space, such that the complementary subspace is a completely entangled subspace, then the set forms a UPB.} \cite{Bennett99, Bravyi04}. This is given by- $\{\ket{0}\ket{1}\ket{-}, \ket{1}\ket{-}\ket{0}, \ket{-}\ket{0}\ket{1}, \ket{+}\ket{+}\ket{+}\}$, where $\ket{\pm} = \ket{0}\pm\ket{1}$. It is well-known that the complementary subspace corresponding to the space spanned by the states of the UPB is a completely entangled subspace \cite{Bennett99, Parthasarathy04}. Here such an entangled subspace contains the GHZ state $\ket{000}-\ket{111}$, since, it is orthogonal to the states of the UPB. We next assume that $\rho_1 = |G_3\rangle\langle G_3|$, $\ket{G_3} = \ket{000}-\ket{111}$. We also assume that $\rho_2$ is a normalized projection operator onto the subspace spanned by the states of the three-qubit UPB. Then, we fix $m \geq 3$. For simplicity, we construct $\{\sigma_1, \sigma_2\}\equiv\{\rho_1^{\otimes3}, \rho_2^{\otimes3}\}$. $\{\rho_1, \rho_2\}$ is know to be a locally indistinguishable set in many-copy scenario \cite{Bandyopadhyay11} and hence, $\{\sigma_1, \sigma_2\}$ is also locally indistinguishable. Here we fix $\ket{\xi}$ as $\ket{00}+\ket{11}$, a minimum dimensional resource state, sufficient to distinguish $\{\rho_1, \rho_2\}$. This can be understood in the following way. Using the resource state $\ket{\xi}$, it is possible to teleport a qubit to location of another qubit creating a bipartition. Now in that bipartition the given UPB can be extended to a full basis and such a full basis is locally measurable in that bipartition \cite{Bennett99}. This is also equivalent to the local distinguishability of $\{\rho_1, \rho_2\}$ in that bipartition. In this regard, see also Ref.~\cite{Halder24}. However, after discrimination, the parties are left with $\rho_1^{\otimes2}$ with probability $\frac{1}{2}$. Notice that it is never possible to get $\rho_1^{\otimes2}$ from $\ket{\xi}^{\otimes2}$ here, not even probabilistically but the opposite conversion is possible. This is how it is possible to recover the resource state on average.
\end{example}

\section{Discussions}
We have studied here the limitations and advantages of using entanglement batteries and entanglement catalysts in local state discrimination problems. We have shown that it is never possible to distinguish the states of an orthonormal bipartite basis perfectly under LOCC using an (exact) entanglement battery or an (exact) entanglement catalyst if the basis contains only entangled states. Then, we have constructed several nontrivial cases of local state discrimination problems where exact and approximate entanglement batteries or entanglement catalysts can provide advantages. However, there are several differences that we have obtained if we compare local state discrimination problems with the local state transformation problems. For example, a maximally entangled state can be used as a catalyst in local state discrimination problems but this is not true for local state transformations. Furthermore, here a catalyst state does not have any correlation with the rest of the system unlike the local state transformation problems. We also mention that in the cases presented, the exact form of any catalyst state can always be obtained from the finally released entangled state locally.

\textit{Acknowledgement}-- This work was supported by the National Science Centre Poland (Grant No. 2022/46/E/ST2/00115 and 2024/55/B/ST2/01590).

\bibliographystyle{apsrev4-2}
\bibliography{ref}
\end{document}